\documentclass[conference]{IEEEtran}
\usepackage{cite}

\usepackage[cmex10]{amsmath}
\usepackage{amsfonts,amssymb, amsthm}
\usepackage{thmtools,thm-restate}
\usepackage{multicol}
\usepackage{float,graphicx}
\usepackage{caption,subcaption}
\usepackage{url}
\usepackage[colorinlistoftodos]{todonotes}
\usepackage[skins,breakable]{tcolorbox}
\usepackage{stmaryrd}
\usepackage{algorithm}
\usepackage[capitalise,noabbrev,nameinlink]{cleveref}

\newif \ifllncs

\newcommand{\dan}[1]{}

\newcommand{\bool}{\{0,1\}}

\ifllncs \else
\newtheorem{definition}{Definition} 

\newtheorem{theorem}{Theorem}
\newtheorem{lemma}{Lemma}
\newtheorem{corollary}{Corollary}

\fi

\newcommand{\A}{\mathcal{A}}

\newcommand{\ignore}[1]{}

\newcommand{\s}{\sigma}
\newcommand{\la}{\leftarrow}
\def\blackslug
{\hbox{\hskip 1pt\vrule width 8pt height 8pt depth 1.5pt\hskip 1pt}}
\def\qed{\hspace*{\fill}\quad\blackslug\lower 8.5pt\null\par}

\newcommand{\poly}{\sf{poly}}

\newenvironment{proofsketch}{\noindent{\em Proof sketch.~~}}{\qed}

\newlength{\protowidth}

\newlength{\gameheight}

\makeatletter

\newcommand{\secparam}{\lambda}
\newcommand{\csecparam}{\kappa}

\renewcommand{\paragraph}[1]{{\smallskip \noindent {\bf #1}~}}

\newcommand{\negl}{{\sf negl}}
\renewcommand{\poly}{{\sf poly}}

\newcommand{\FACTS}{\textsf{FACTS}}
\newcommand{\mc}[1]{\mathcal{#1}}
\newcommand{\setup}{\textsf{Setup}}
\newcommand{\sendee}{\textrm{send}}
\newcommand{\receiveee}{\textrm{receive}}
\newcommand{\send}{\textsf{SendMsg}}
\newcommand{\receive}{\textsf{RcvMsg}}
\newcommand{\complain}{\textsf{Complain}}
\newcommand{\audit}{\textsf{Audit}}
\newcommand{\orig}{\textsf{Originate}}

\DeclareMathOperator{\Enc}{\textsf{Enc}}
\DeclareMathOperator{\Dec}{\textsf{Dec}}
\DeclareMathOperator{\Sig}{\textsf{Sig}}
\DeclareMathOperator{\Ver}{\textsf{Ver}}

\renewcommand{\tag}{\mathsf{tag}}
\newcommand{\tagg}{\mathsf{tag}}

\def\eem{\mathsf{EEMS}}

\def\cp{\mathsf{user\mbox{-}privacy}}
\def\sp{\mathsf{server\mbox{-}privacy}}
\def\uf{\mathsf{unforgeability}}
\def\nd{\mathsf{no-delay}}
\def\nsu{\mathsf{no-speedup}}
\def\GameSP{\mathsf{Game}_\eem^{\sp}(\mc{A})}
\def\AdvSP{\mathsf{Adv}_\eem^{\sp}(\mc{A})}
\def\GameCP{\mathsf{Game}_\eem^{\cp}(\mc{A})}
\def\AdvCP{\mathsf{Adv}_\eem^{\cp}(\mc{A})}
\def\GameUF{\mathsf{Game}_\eem^{\uf}(\mc{A})}

\def\GameND{\mathsf{Game}_\eem^{\nd}(\mc{A})}

\def\GameNSU{\mathsf{Game}_\eem^{\nsu}(\mc{A})}

\def\Expect{\mathbb{E}}

\def\USED{\mathrm{USED}}

\newcommand{\falling}[2]{\ensuremath{{#1}^{\underline{#2}}}}

\IEEEoverridecommandlockouts
\makeatletter\def\@IEEEpubidpullup{6.5\baselineskip}\makeatother
\IEEEpubid{\parbox{\columnwidth}{
    Network and Distributed Systems Security (NDSS) Symposium 2022\\
    23-26 February 2020, San Diego, CA, USA\\
    ISBN 1-891562-61-4\\
    https://dx.doi.org/10.14722/ndss.2020.23xxx\\
    www.ndss-symposium.org
}
\hspace{\columnsep}\makebox[\columnwidth]{}}

\title{Fighting Fake News in Encrypted Messaging with the Fuzzy Anonymous Complaint Tally System (FACTS)}
\author{
\IEEEauthorblockN{Linsheng Liu\IEEEauthorrefmark{1}, Daniel S.\ Roche\IEEEauthorrefmark{2}, Austin Theriault\IEEEauthorrefmark{1}, Arkady Yerukhimovich\IEEEauthorrefmark{1}}
\IEEEauthorblockA{\IEEEauthorrefmark{1}George Washington University \emph{lls,atheriault,arkady@gwu.edu}}
\IEEEauthorblockA{\IEEEauthorrefmark{2}United States Naval Academy \emph{roche@usna.edu}}
}

\begin{document}

\maketitle
\begin{abstract}
    Recent years have seen a strong uptick in both the prevalence and real-world consequences of false information spread through online platforms. At the same time, encrypted messaging systems such as WhatsApp, Signal, and Telegram, are rapidly gaining popularity as users seek increased privacy in their digital lives. 
    The challenge we address is how to combat the viral spread of misinformation without compromising privacy. Our FACTS system tracks user complaints on messages obliviously, only revealing the message's contents and originator once sufficiently many complaints have been lodged.
    Our system is \emph{private}, meaning it does not reveal anything about the senders or contents of messages which have received few or no complaints; \emph{secure}, meaning there is no way for a malicious user to evade the system or gain an outsized impact over the complaint system; and \emph{scalable}, as we demonstrate excellent practical efficiency for up to millions of complaints per day.
    Our main technical contribution is a new \emph{collaborative counting Bloom filter}, a simple construction with difficult probabilistic analysis, which may have independent interest as a privacy-preserving randomized count sketch data structure.
    Compared to prior work on message flagging and tracing in end-to-end encrypted messaging, our novel contribution is the addition of a high \emph{threshold} of multiple complaints that are needed before a message is audited or flagged.
    We present and carefully analyze the probabilistic performance of our data structure, provide a precise security definition and proof, and then measure the accuracy and scalability of our scheme via experimentation.
\end{abstract}
\section{Introduction}
The proliferation of fake and misleading information online has had significant impact on political discourse~\cite{NYTimes:IsaRoo18} and has resulted in violence~\cite{WAPo:Samuels20}. Large services like Facebook and YouTube have begun to remove or label content that they know to be fraudulent or misleading~\cite{Facebook:fakeNews, Youtube:fakeNews}, through a combination of a manual process of reviewing posts/videos and automated machine learning techniques.

However, on end-to-end encrypted messaging services (EEMS), like Signal, WhatsApp, Telegram, etc., where so-called ``fake news'' is also shared, such review is impossible. At no point do the providers see the plain-text, unencrypted contents of messages transmitted through their systems and thus cannot identify and remove offending material.  Such platforms must instead rely on their users to identify and report malicious content.  Even then, identifying and removing \emph{users} who repeatedly post misleading and dangerous content may still be difficult because some platforms, like Signal, also hide the path the message took, so identifying and addressing the original source of the misinformation may not be possible.

Tyagi et al.~\cite{CCS:TyaMieRis19} introduced a first approach for overcoming this challenge and allow EEMS to effectively \emph{traceback} an offending  message to find the originator based on a user complaint.  The traceback procedure also assures that all other messages remain private and that innocent parties cannot be blamed for originating the offending messages.

While innovative, there are two notable shortcomings of Tyagi et al.'s traceback scheme. First, it requires extensive ``housekeeping'' on the part of the platform that scales as the number of messages in the system. Second, a single, possibly malicious, complaint can trigger a traceback and thus reveal the message contents as well as the history of prior recipients, which is counter to  the goals of EEMS to maintain the privacy of users communicating through this system. One malicious user (e.g., a government agent) can reveal the source of a piece of information (e.g., a leak) that they have received, violating the privacy of the sender (e.g., the leaker) by issuing a single complaint to the EEMS. While it may be possible to apply manual review to these complaints, the scale of possible complaints could make this impractical. 

In this paper, we aim to resolve this conflict between privacy and ability to identify misinformation in EEMSs by first observing that ``fake news'' messages are, by definition, viral and are thus received, and likely complained about, by a large number of users.  Private messages, such as leaks, on the other hand, are likely to be targeted and are thus only received by a small number of users; indeed, any message received by only a few users is inherently less impactful overall and more likely deserving of privacy protections.  This leads to a more nuanced approach for identifying fake news: apply a threshold approach to complaint management, whereby only viral fake news would overcome the threshold and trigger an audit. 

Counting the number of complaints in a private manner is a non-trivial problem if the privacy of the EEMS' clients is to be maintained prior to the threshold being reached, even given available cryptographic solutions. For example, a homomorphic encryption solution (e.g.,~\cite{EC:KatMyeOst01}) would enable the checking and updating of counts for each message, but the access patterns of clients checking and updating counters could reveal how many complaints a message receives even if the threshold is not reached. Oblivious RAM  (ORAM) (e.g.~\cite{JACM:GolOst96, JACM:SDSCFRYD18}) could be used to protect the access patterns, but have high computational overheads and usually assume clients may share secrets and are not malicious.
Pricate Information Retrieval (PIR) does not assume clients are trusted, but has different scalability challenges and does not address the problem of obliviously updating without revealing which message is being complained about.

We propose a different approach we call a  Fuzzy Anonymous Complaint Tally System (FACTS).  FACTS maintains an (approximate) counter of complaints for each message, while also ensuring that, until a threshold is exceeded, the status of these counters is kept private from the server and all users who have not received the message.  FACTS builds on top of any end-to-end encrypted messaging platform, incurring only small overhead for message origination and forwarding.  In particular, FACTS maintains the communication pattern of the underlying messaging system, requiring no new communication or secrets between users even for issuing complaints.   

To avoid the high overheads of existing solutions, FACTS uses a novel oblivious data structure we call a \emph{collaborative counting Bloom filter} (CCBF).  This data structure allows us to obliviously increment and query approximate counters on millions of messages while only requiring 12MB of storage.  Moreover, incrementing a counter only requires flipping \emph{one bit} on the server and only uses the minimal communication of $\log{|T|}$ bits to address a single bit in the server-stored bit vector $T$.  While the resulting counters are only approximate, we show experimentally and analytically that we are able to enforce the threshold on complaints with good accuracy, namely,  below 10\% error in theory, and below 3\% in most realistic deployment scenarios.   

The contributions of this paper are as follows:
\begin{itemize}
    \item We develop a collaborative counting Bloom filter, a new oblivious data structure for counting occurrences of a large number of distinct items.
    \item We use this data structure to instantiate a provably-secure system, FACTS, for privacy-preserving source identification of fake news in EEMSs.
    \item Finally, we perform experiments to show the accuracy and overhead of FACTS in realistic deployment scenarios.
\end{itemize}

\subsection{Setting and Goals}
FACTS is built on top of an 
end-to-end encrypted messaging system (EEMS).  For this work, we focus on the setting of server-based EEMSs with a server $S$ that enables (authenticated) encrypted communication between the system users.  Examples of such EEMSs include Signal and WhatsApp, among many others.  

To make sure that FACTS is compatible with existing encrypted messaging systems, we make the following performance requirements:

\begin{enumerate}
    \item \textbf{Messaging costs:}  Originating and forwarding messages should incur little computational overhead for both users and the server over the standard procedure in the encrypted messaging system,
    \item \textbf{Server storage:} The storage overhead of the server should be small (i.e., a single table not exceeding a few MBs),
    \item \textbf{User costs and requirements:}
    Issuing complaints requires a small amount of communication and computation from the complaining user, and no cost to other users.  Moreover, complaints can not require direct communication between users or require the users to have any apriori shared secrets that are not known to the server.
    \item \textbf{Complaint throughput:}  Issuing complaints may be slower than standard forwarding of messages, but the system must be able to handle millions of complaints per day. 
\end{enumerate}

To ensure privacy of messages and complaints, FACTS requires that complaints remain hidden from the server (and colluding clients) until a threshold of complaints is reached.  Additionally, FACTS ensures integrity of the complaint process ensuring correctness of complaint counts and the identity of the revealed originator once the threshold is reached. Specifically, FACTS satisfies the following security guarantees:

\begin{enumerate}
    \item \textbf{Message privacy:} All messages remain end-to-end encrypted and private from the server and non-receiving clients until a threshold of complaints is reached and an audit is issued.  Moreover, even after the audit, only information about the audited message is revealed.  
    
    \item \textbf{Originator integrity:}  Once a threshold of complaints is reached on a message, FACTS will only identify information about the true originator of the message.  In particular, no innocent party can be framed as the originator.
    
    \item \textbf{Complaint privacy:} The server and any colluding clients who have not received a message $x$ should have no information about the number of complaints on $x$.  In particular, the server should not be able to tell what message is being complained about.
    
    \item \textbf{Complaint integrity:}  A set of malicious clients should not be able to alter the number of complaints on any message $x$.  Specifically, they cannot block or delay complaints, and cannot (significantly) increase the number of complaints on a message $x$ except through the legitimate complaint process.
\end{enumerate}


\subsection{Building FACTS}
Recall that our goal is to enable privacy-preserving counters to tally complaints on each message $m$.  This suggests an immediate solution where the server stores an encrypted counter for each message, and clients interact with the server to increment the counter and check the threshold.  While implementing such counters is certainly possible using homomorphic encryption~\cite{STOC:Gentry09} or standard secure computation techniques~\cite{FOCS:Yao86, STOC:GolMicWig87, STOC:BenGolWig88} , the problem is that the access pattern of clients' updates to counters leaks information to the server by revealing the complaint histogram.  
This suggests a further modification to store the counters inside an oblivious RAM (ORAM)~\cite{JACM:GolOst96} to hide such access patterns from the client.  However, in our setting this would require a multi-client ORAM~\cite{maffei-msmcoram,concuroram,MOSE} which incurs significant performance penalties including at least $O(\log{n})$ communication overhead when there are $n$ distinct messages.  Moreover, this would require direct communication between clients to maintain their ORAM state, and additionally, no security against malicious clients.

In FACTS, we take a different approach.  Instead of relying on encryption to hide the counters from the server, we hide the counters in plain sight by mixing together the counters for all the messages in a way oblivious to the server.  To make this possible, we relax the functionality of FACTS to only enforce approximate, rather than exact, thresholds.  That is, the threshold will be triggered on a message $x$ after $(1\pm \epsilon) t$ complaints for a small error $\epsilon$.
Making this relaxation allows us to use a sketch-based approach for counting the complaints.  

To achieve this functionality obliviously, we develop a collaborative counting Bloom filter (CCBD).
This data structure consists (roughly) of a collection of Bloom filters, one for each message, where the Bloom filters corresponding to different messages are mixed together to hide them from the server. Specifically, the server stores a table of $s$ bits.  A random subset of $v$ bits ($V_x$) is assigned to each message $x$ at origination; these bits will be used for tracking complaints about this message (for intuition, one can think of these bits as forming a Bloom filter for storing the set of complaints about the message).  We stress that the server has no information about which bits correspond to which messages. 

To complain about a message $x$, a user who has received $x$ can find the corresponding bit locations, and will (attempt to) flip one of the bits from 0 to 1.  However, allowing users to flip any bit they choose, would allow malicious users to significantly accelerate complaints for a message they wish to disclose.  To prevent this behavior, we restrict each client to only be able to flip (i.e., complain on) a small (of size $u$) set of locations $U_C$.  Thus, to complain about a message $m$, a client first identifies the set $V_x$ of bits corresponding to $x$.  Then, she checks how many of these bits have already been set to 1, and if this exceeds a specified threshold, notifies the server to trigger an audit.  If the threshold for $x$ is not yet exceeded, the client sees whether any of the 0 bits in $V_x$ are in her set $U_C$, and if there are any such bits, she flips one of them (chosen at random) from 0 to 1.  Otherwise, the user still flips a random bit in their set $U_x$, so the server cannot discern anything about the message being complained on. We prove in Section~\ref{sec:ccbf} that the actual number of complaints necessary to trigger an audit can be calculated with high precision allowing us to (approximately) enforce the desired threshold.

\subsection{Limitations of FACTS}
In order to present FACTS, it is also important to recognize what our system does \emph{not} do.

First, unlike some prior work~\cite{ICS:Geiger16},%
\cite{km-private-membership-usenix21}, FACTS does not attempt to automatically detect misinformation.  Instead, it relies on users reporting it when they see it.
This reliance on users has inherent benefits and limitations.
While our system is not subject to the kinds of machine-generated false positives
that can arise from, e.g., hash collisions
\cite{brandom-verge-apple-csam},
our model is inherently vulnerable to any sufficiently large group of dishonest users, who could trigger an audit on a benign message.
This is why we suggest the possibility of a manual human review process on message contents before the service provider would take any action on an audited message; see \cref{sect:third-party-audits}.

Second, due to the approximate nature of FACTS, it works most effectively for relatively large thresholds, say in the hundreds and above.  For our application to fake news detection, this is reasonable as such messages are likely to garner a large number of complaints, and indeed this was our main motivation for this paper. We leave as interesting possible future work to implement a system supporting smaller thresholds, even as small as 2, efficiently.

One additional functionality limitation is that, as is true with any application using Bloom filters, the CCBF data structure can fill up once too many complaints have been registered.  To deal with this issue it is necessary to periodically reset the counters and refresh the CCBF data structure.  We refer to each such refresh period as an \emph{epoch}, and in the remainder of the paper only present algorithms for a single epoch.

Finally, on the security side, an important limitation is that FACTS reveals meta-data on who issues complaints (but not what message they complain on).  It is important to consider what is revealed by this meta-data.  By observing the timing of messages and complaints, the server can make some inferences about what messages users are sending and complaining about.  For example, suppose that the server sees that $A$ sends a message to $B$, and then $B$ issues a complaint. Then, it may be reasonable for the server to assume that $A$ has sent the message which $B$ complained about, even though this is not directly leaked by our system. Nonetheless, our definition guarantees that the server cannot be certain that this is indeed the case.  We note that the messaging meta-data is already a byproduct of the underlying EEMS platform.  FACTS only adds complaint meta-data to this leakage; see \cref{sec:altFACTS} for some further discussion.

\subsection{Paper Layout}
The remainder of the paper is organized as follows.  In Section~\ref{sec:prelim}, we introduce some of the notation we use throughout the paper.  Then, in Section~\ref{sec:facts} we describe the syntax and functionality of FACTS.  In Section~\ref{sec:ccbf} we present and analyze our main building block, the CCBF data structure.  Then, in Section~\ref{sec:construction}, we show how to use a CCBF to instantiate FACTS.  We demonstrate the accuracy and performance through experimental evaluation in Section~\ref{sec:experiments} and then prove the security of FACTS in Section~\ref{sec:proofs}.  Finally, we describe some variants of FACTS and directions for future work in Section~\ref{sec:altFACTS} and present related work in Section~\ref{sec:related}.

\section{Preliminaries} \label{sec:prelim}
We use $[n]$ to denote the set $1, \ldots, n$.  We write $x \gets X$ to indicate that the value $x$ is sampled uniformly at random from the set $X$.  We use $\secparam$ to denote a statistical security parameter and $\csecparam$ to denote a computational security parameter.  We also assume the existence of a hash function $H:\{0,1\}^* \rightarrow \{0,1\}^*$ which is modeled as a random oracle.  We let $\poly(\cdot)$ denote a polynomial function and $\negl(\cdot)$ denote a negligible function.  


\section{Fuzzy Anonymous Complaint Tally System (FACTS)}\label{sec:facts}

In this section, we present the syntax for FACTS and describe how FACTS is used.  We show how to instantiate FACTS in Section~\ref{sec:construction}.

\paragraph{Assumptions:}
We assume that each user $A$ has a unique identifier $ID_A$, and that the server can authenticate these IDs.  (We will abuse notation to use $A$ to represent the user and also the id $ID_A$).  We also assume that the server has an identifier $ID_S$ (we will denote this by $S$) that can be authenticated by all users.

Additionally, we assume that the underlying end-to-end encrypted messaging system (EEMS) offers methods $\sendee(A,B,x)$ and $\receiveee(A,B,x)$ for sending and verifying a message $x$ sent from user $A$ to user $B$.  Moreover, we assume that this communication is encrypted and authenticated. In particular, $\receiveee$ verifies that the received message was sent by $A$ and was not modified in transit.  Importantly, we do not assume that this platform is anonymous, instead assuming that the full messaging history i.e., who sent a message to whom and the size of that message, is available to the server. 

\paragraph{Syntax:}
FACTS is a tuple of protocols $\FACTS=(\setup, \send, \receive, \complain, \audit)$.  The first is used to set up FACTS, the next two are used to send and verify messages, while the last two methods are used to issue complaints and audit received messages.

\begin{itemize}
    \item $\setup(c)$:  This takes as input the total number of users and initiates the FACTS scheme for $c$ users.
    
    \item $\send(A, B,\tag_x,x)$: This method is used by a user $A$ to send a message $x$ to another user $B$.  This may be a new message \emph{originated} by $A$ (indicated by $\tag_x = \perp$) or a forward of a previously received message.

    \item $\receive(A,B,\tag_x, x)$:  This algorithm is run by $B$ upon receiving a message $(\tag_x, x)$ from $A$.
    
    This algorithm checks whether $\tag_x$ is indeed a valid tag generated by $A$ on message $x$.  If this is the case, then $B$ accepts the message, otherwise he rejects the message.
    
    \item $\complain(C, \tag_x, x)$:  This protocol is run by a user $C$ to complain about a received message $(\tag_x,x)$.  
    
    \item $\audit(C,\tag_x, x)$:  This protocol issues an audit of a message $x$ revealing $(\tag_x, x)$ to $S$.  This will be called by $C$ when the number of complaints on $m$ exceeds a pre-defined threshold (with high probability).
\end{itemize}

\paragraph{Usage:}
The following workflow demonstrates the standard usage of FACTS.  To originate a new message $x$, a user $A$ runs the $\send$ protocol with the server $S$ to create metadata $\tag_x$.  $\send$ then sends this metadata and the message $(\tag_x,x)$
to the receiving user $B$ using the messaging platforms $\sendee$ method.  Upon receiving a message $(\tag_x, x)$, $B$ first locally runs $\receive(A,B,\tag_x,x)$ to verify that the received message and tag are valid, if this fails he ignores the message.  To forward a received message $(\tag_x,x)$, a user $A$ runs $\send$ with the server $S$ to produce metadata $\tag'_x$; $A$ then discards this metadata, and the original message $(\tag_x, x)$ is sent instead using the messaging platform's $\sendee$ method.\footnote{We note that since the underlying messaging scheme is encrypted, the actual ciphertext sent will not be the same as the ciphertext received.}

If a user $B$ receives a message $(\tag_x, x)$ that it considers ``fake'', he can use the  $\complain$ protocol to issue a new complaint on this message.  After issuing a complaint, $B$ checks whether the threshold of complaints on $x$ has been reached.  If so, he calls $\audit$ to trigger an audit on the message $(\tag_x,x)$, revealing $x$ and the originator of $x$ to the server $S$.


\section{Collaborative counting Bloom filter}\label{sec:ccbf}

Our system records complaints in a special data structure which we
call a \emph{collaborative counting Bloom filter}, or CCBF. This data structure
shares some of the same basic functionality as a counting Bloom filter~\cite{TON:FCAB00, PODC:Mitzenmacher01} or count-min sketch~\cite{JAL:CorMut05}, which is to insert elements and
compute the (approximate) frequency of a given element.

Our CCBF differs from a usual count-min sketch in that each update
operation is accompanied by a \emph{user id}, and each user can only perform
a single update for a given element. This can be thought of as a strict generalization
of the normal count-min sketch operations, where the latter may be simulated by our
CCBF by choosing a unique user id for each update.

The actual data structure for the CCBF is also far simpler than the 2D array
of integers used for a count-min sketch; instead, we store only a single length-$s$
bit vector $T$. As a result, our CCBF will have the following desirable properties:

\begin{itemize}
    \item The bit-length of $T$ scales linearly with the total
    number of insertions.
    \item Each witness operation (insertion) changes exactly one bit in the underlying
    bit vector from 0 to 1.
    \item The CCBF is \emph{item-oblivious}, meaning that after observing
    an interactive update protocol, the adversary learns which user id made the update,
    but not which item was updated.
\end{itemize}

The downside to our CCBF is a far lower accuracy of the count operation in general compared to count-min sketches.
However, we will show that, for careful parameter choices, the count operation is highly
accurate within a certain range, which is precisely what is needed for the
current application.

\subsection{CCBF Construction}

The CCBF consists of a single size-$s$ bit vector $T$ and two operations:

\newcommand{\Increment}{\ensuremath{\mathsf{Increment}}}
\newcommand{\TestCount}{\ensuremath{\mathsf{TestCount}}}

\begin{itemize}
    \item $\Increment(x, C)$: Increases the count by 1 for item $x$
    according to user id $C$.
    \item $\TestCount(x, t)$: Returns true if the number of increments performed so far for item $x$ is \emph{probably} greater than or equal to $t$.
\end{itemize}

Note that $\TestCount$ is probabilistic, in the sense that it may return false
when the actual count is greater than $t$, or true when the actual count is less
than $t$. Our construction guarantees the correctness probability is always at least
$\tfrac{1}{2}$, and our tail bounds below show the correctness probability quickly
goes towards $1$ when the actual count is much smaller or larger than $t$.

The performance and accuracy of the CCBF is governed by three integer parameters
$s$, $u$, and $v$, with $u,v\le s$,
which must be set at construction time. The first, $s$, is the fixed size of the
table $T$. Each user $i$ is randomly assigned a static set of exactly $u$ locations in
the $T$; i.e., a uniformly random subset of $\{0,1,2,\ldots,s-1\}$, which we
call the \emph{user set}.
Similarly, each possible item $x$ is assigned a random set of exactly $v$
bit vector locations, which we call the \emph{item set}.

The two CCBF operations can be implemented by a single server and any number of clients. The protocols are simple and straightforward, save for the calculation of the \emph{tipping point} $\tau$ which we present in the next subsection.

In these protocols, the size-$s$ bit vector $T$ is considered \emph{public} or \emph{world-readable}; it is known by all parties at all times. In reality, the server who actually stores $T$ may send it to the client periodically, or whenever a client initiates a \Increment{} or \TestCount{} protocol. However, the bit vector $T$ is only writable by the server.

The $\Increment(x,C)$ protocol, outlined in \cref{alg:increment}, involves the User attempting to set a single bit from 0 to 1 within the item set for $x$. However, the user is only allowed to write locations within their own user set. So, if there are no 0 bits in the intersection of these two index sets, the user instead changes any other arbitrary 0 bit in its own user set in order to maintain item obliviousness.

\begin{algorithm}[htp]
    \caption{$\Increment(x,C)$\label{alg:increment}}
\begin{enumerate}
    \item User and server separately compute the list of $u$ user locations for user $C$, $U_C \subseteq \{0,\ldots,s-1\}$.
    \item User computes list of $v$ item locations for item $x$,
   $V_x \subseteq \{0,\ldots,s-1\}$
   \item User checks each location in $U_C$ in the table $T$ to compute a list $S_C = \{i\in U_C \mid T[i] = 0\}$ of \emph{settable} locations for user $C$
   \item If $S_C = \emptyset$, then the user cannot proceed and calls \textbf{abort}.
   \item Else if $S_C \cap V_x \ne \emptyset$, user picks a uniformly random index $i\gets S_C \cap V_x$ and sends index $i$ to server.
   \item Else user picks a random index $i\gets S_C$ and sends index $i$ to server.
   \item Server checks that received index $i$ is in the user set $U_C$ and that $T[i] = 0$, then sets $T[i]$ to 1.
\end{enumerate}
\end{algorithm}

Since the bit vector $T$ is considered world-readable, the only communication here is the single index $i$ from client to server over an authenticated channel. In reality, to avoid race conditions, the server will actually send the table entry values $T[i]$ for all $i\in U_C$ to the user first and lock the state of the global bit vector $T$ until receiving the single index response back from the user.

The $\TestCount(x,t)$ protocol is not interactive as it only requires reading the entries of $T$. The precise computation of the \emph{tipping point} $\tau$ is detailed in the next section. Note that this computation depends only on the \emph{total} number of bits set in the bit vector $T$ as well as the parameters $s,u,v$; therefore the computation of $\tau$ is independent of the item $x$ and could for example be performed once by the server and saved without violating item obliviousness.

This protocol is detailed in \cref{alg:testcount}.

\begin{algorithm}[htp]
    \caption{$\TestCount(x,t)$\label{alg:testcount}}
\begin{enumerate}
    \item Use parameters $s,u,v$ and current value of $m$ total number of bits set in $T$, to compute the tipping point $\tau$.
    \item Compute list of $v$ item locations for item $x$, $V_x \subseteq \{0,\ldots,s-1\}$
    \item Check how many bits of $T$ are set for indices in $V_x$. Return \textbf{true} if and only if this count is greater than or equal to $\tau$.
\end{enumerate}
\end{algorithm}

\subsection{Calculating the tipping point}

The key to correctness of the \TestCount{} protocol is a calculation of the \emph{tipping point} $\tau$, which is the expected number of 1 bits within any item set, if that item has been incremented $t$ times. We now derive an algorithm to compute this expected value exactly, in $O(tv)$ time and $O(v)$ space.

Let $s$ be the total size of the table $T$ and $m\le s$ be the total number of calls to \Increment{} so far. That is, $m$ equals the number of 1 bits in $T$.
Recall that $u,v \le s$ are the number of table entries per user and per item, respectively.

We first derive the probability that two subsets of the $s$ slots have given-size intersection.
Next we derive a recursive formula for $\tau$ using these intersection probabilities.
The nearest integer to $\tau$ can then be efficiently computed using a simple dynamic
programming strategy.

\paragraph{Intersection probabilities}

For the remainder, we use Knuth's notation $\falling{n}{k}$ to denote the \emph{falling factorial},
defined by
\[\falling{n}{k} = \frac{n!}{(n-k)!} = n\cdot(n-1)\cdot(n-2)\cdots(n-k+1).\]

\begin{lemma}\label{lem:intersect}
Let $k,a,b,s$ be non-negative integers with $k \le b \le a \le s$, and
suppose $S$ and $T$ are two subsets of a size-$s$ set with $|S|=a$ and $|T|=b$, each
chosen independently and uniformly over all subsets with those sizes. Then
\begin{equation}\label{eqn:intersect}
\Pr(|S \cap T| = k) = 
    \frac{\falling{a}{k} \cdot{} \falling{b}{k} \cdot{} \falling{(s-a)}{b-k}}
         {\falling{s}{b} \cdot{} k!}.
\end{equation}
\end{lemma}
\begin{proof}
The number of ways to choose $S$ and $T$ with a size-$k$ intersection, divided
by the total number of ways to choose two size-$a$ and size-$b$ sets, equals
\[\frac{\binom{s}{k} \cdot{} \binom{s-k}{a-k} \cdot{} \binom{s-a}{b-k}}
{\binom{s}{a}\cdot{}\binom{s}{b}}.\]
This simplifies to \eqref{eqn:intersect}.
\end{proof}

Because the numerator and denominator are each products of $b+k$ single-precision integers,
the value of \eqref{eqn:intersect} can be computed in $O(b)$ time to full accuracy in
machine floating-point precision.

Furthermore, equation \eqref{eqn:intersect} has the convenient property that, after altering any
value $a$, $b$, or $k$ by $\pm 1$, we can update the probability with only $O(1)$ additional
computation. So, for example, one can compute the probabilities for every $k\le b$ in the same
total time $O(b)$.

\paragraph{Recurrence for number of unfilled message slots}

Fix an arbitrary item $x$, and let $w\le v$ denote the number of 0 bits of $T$ within $x$'s item set. Let $k\le m$ denote the number of \Increment{} operations performed on item $x$ performed so far.

First, for convenience define $p_w$ to be the probability that an arbitrary user is able to
write to one of the $w$ remaining unfilled slots for the message. From \cref{lem:intersect}, we have
\begin{equation}
    p_w = 1 - \frac{\falling{(s-u)}{w}}{\falling{s}{w}},
    \label{eqn:pw}
\end{equation}
which can be computed in $O(w)$ time. In fact, we pre-compute \emph{all} possible values of $p_w$
with $0\le w\le v$ in $O(v)$ total time.

Now consider the random variable for the number of 0 bits within $x$'s item set after $k$ \Increment{}'s on $x$, if the item set originally had $w$ 0 bits.
Define $R_{w,k}$ to be the expected value of this random variable, which can be calculated recursively as follows.

If $w=0$, then the slots are all filled, and if $k=0$ then there are no more \Increment{}'s, so
the number of unfilled slots remains at $w$. Otherwise, the first \Increment{} will fill an additional slot
with probability $p_w$, leaving $w-1$ remaining unfilled slots, and otherwise will leave
$w$ remaining unfilled slots. This implies the following recurrence relation:

\[
R_{w,k} = \begin{cases}
    0,& w=0 \\
    w,& k=0 \\
    p_w R_{w-1,k-1} + (1 - p_w) R_{w,k-1},& w,k\ge 1
\end{cases}
\]

All values of $R_{w,t}$ with $0\le w \le v$ can be computed in $O(tv)$ time and $O(v)$ space,
using a straightforward dynamic programming strategy.

\paragraph{Computing the tipping point}

We now show how to compute the tipping point value $\tau$, which is the expected number of filled item slots
after $t$ \Increment{}s on that item, by summing the $R_{w,t}$ values over all possible
values of $w$ based on the number of \emph{other} calls to \Increment{} $m$.

To this end, define $q_w$ to be the probability that $w\le v$ slots for a given item are unfilled
after $m$ total calls to \Increment{} for other items. Because other calls to \Increment{} are for other
unrelated items, each one goes to a uniformly-random unfilled slot over the entire size-$s$ table $T$.
Therefore $q_w$ is the same as the probability of a size-$m$ set and a size-$v$ set having intersection size exactly $v-w$. From \cref{lem:intersect}, this is
\[q_w = 
    \frac{\falling{m}{v-w} \cdot{} \falling{v}{v-w} \cdot{} \falling{(s-m)}{w}}
         {\falling{s}{v} \cdot{} (v-w)!}.
\]
We can pre-compute all values of $q_w$ for $0\le w\le v$ in total time $O(v)$.

After pre-computing the values of $p_w$, $R_{w,t}$, and $q_w$, we can finally express
the tipping point $\tau$ as a linear combination
\begin{equation}\label{eqn:r}
    \tau = v - \sum_{w=0}^v q_w R_{w,t},
\end{equation}
rounded to the nearest integer.

In total, the computation requires $O(tv)$ time and $O(v)$ space.

\subsection{Tail Bounds}\label{sec:tailbounds}

\newcommand{\CminvOverk}{7.042652}
\newcommand{\CminukOvers}{0.5184846}
\newcommand{\Cplow}{0.956414}
\newcommand{\CmaxvOvers}{0.00386}
\newcommand{\CmaxuvOvers}{3.65151}
\newcommand{\Cphigh}{0.974876}
\newcommand{\CzeroUpper}{3.66567}
\newcommand{\CvLower}{371}
\newcommand{\CboxUpper}{0.00989}
\newcommand{\CminsOvern}{96}
\newcommand{\CmaxvOvert}{7.409}
\newcommand{\CmaxvOvertb}{8}
\newcommand{\Cfour}{1.0520553}
\newcommand{\Cminfp}{2.1}

Next, we prove lower and upper bounds on the probability of filling a single additional item slot during an \Increment{} operation, \cref{lem:lowerp,lem:upperp} respectively.

\begin{restatable}{lemma}{lemmalowerp}\label{lem:lowerp}
    Let $x$ be an item such that at most
    $\tau$ of $x$'s item slots are filled.
    If the CCBF parameters $s,u,v$
    satisfy
    $v \ge \CminvOverk{} \tau$ and
    $u \ge \CminukOvers{} \tfrac{s}{\tau}$,
    then the probability that a call to $\Increment(x,C)$ fills in
    one more of $x$'s item slots is at least
    $\Cplow$.
\end{restatable}

\begin{restatable}{lemma}{lemmaupperp}\label{lem:upperp}
    Let $x$ be any item.
    If the CCBF parameters $s,u,v$ satisfy
    $\CvLower{} \le v \le \CmaxvOvers{} s$ and
    $u \le \CmaxuvOvers{} \tfrac{s}{v}$,
    then the probability that a call to $\Increment(x,C)$ fills in
    one more of $x$'s item slots is at most
    $\Cphigh$.
\end{restatable}

Now we use the probability upper bound to prove an upper bound on the tipping point $\tau$.

\begin{restatable}{lemma}{lemmauppertau}\label{lem:uppertau}
    Let $s,u,v$ be CCBF parameters that satisfy the conditions of \cref{lem:upperp},
    and suppose $m,t$ are integers such that
    $s \ge \CminsOvern{} m$ and
    $v \le \CmaxvOvert{} t$. Then the tipping point $\tau$, for threshold $t$
    and with $m$ total set bits in the table $T$, is at most
    $\Cfour{} t$.
\end{restatable}

We can now state our main theorems on the accuracy of the CCBF data structure. Consider a call to the predicate function $\TestCount(x,t)$, which attempts to determine whether the number of prior $\Increment{}$ calls with the same item $x$ is at least $t$. Our exact computation of the tipping point $r(t)$ shows that this function always returns the correct answer with at least 50\% probability. But of course, so would a random coin flip!

Let $k$ be the \emph{actual} number of calls to $\Increment(x,C)$ that have occurred. Then two kinds of errors can occur: a \emph{false positive} if $\TestCount(x,t)$ returns true but $k<t$, and a \emph{false negative} if $\TestCount(x,t)$ returns false when $k\ge t$. Intuitively, both errors occur with higher likelihood when the true count $k$ is close to $t$. Our main theorem captures and quantifies this intuition, saying that, ignoring low-order terms, $\TestCount$ is accurate to within a 10\% margin of error with high probability.

\begin{restatable}{theorem}{theoremfp}\label{thm:fp}
    Let $n$ be an upper bound on the total number of calls to $\Increment$, and $t$ be a desired threshold for $\TestCount$. Suppose the parameters $s,u,v$ for a CCBF data structure satisfy the conditions of \cref{lem:lowerp}, and furthermore that
    $v \le \CmaxvOvertb{} t$.
    If the actual number of calls to $\Increment(x,C)$ is at most
    $t - \Cminfp\sqrt{\lambda t}$, then the probability $\TestCount(x,t)$ gives
    a false positive is at most $2^{-\lambda}$.
\end{restatable}

\begin{restatable}{theorem}{theoremfn}\label{thm:fn}
    Let $n$ be an upper bound on the total number of calls to $\Increment$, and $t$ be a desired threshold for $\TestCount$. Suppose the parameters $s,u,v$ for a CCBF data structure satisfy the conditions 
    of \cref{lem:lowerp,lem:uppertau}.
    If the actual number of calls to $\Increment(x,C)$ is at least
    \begin{equation} \label{eqn:fn}
        1.1 t + .4 \lambda + .7\sqrt{\lambda t},
    \end{equation}
    then the probability $\TestCount(x,t)$ gives
    a false negative is at most $2^{-\lambda}$.
\end{restatable}

We can easily summarize the various conditions on the parameters as follows:
\begin{corollary}\label{cor:params}
    Let $n$ be a limit on the total number of calls to \Increment{}, and
    $t$ be a desired threshold satisfying 
    $50 \le t \le \tfrac{n}{20}$.
    Then by setting the parameters of a CCBF according to
    $s=96n$, $v=7.409t$, and $u=47.31\tfrac{n}{t}$, any call to
    $\TestCount(x,t)$ will satisfy the high accuracy assurances of
    \cref{thm:fp,thm:fn}.
\end{corollary}

\section{Instantiating FACTS}\label{sec:construction}
We are now ready to present our construction of FACTS.  This construction is based on the collaborative counting Bloom filter (CCBF) data structure presented in Section~\ref{sec:ccbf} to obliviously count the number of complaints on each message.  It uses an underlying EEMS for sending end-to-end encrypted messages between users.

\paragraph{Setup:}
The setup procedure for FACTS first sets up the underlying end-to-end encrypted messaging system (EEMS).  For simplicity, we assume that there is a fixed number $c$ of users using the system.  Setup  generates all necessary keys for the server $S$ and all $c$ users and distributes the keys.  We note that if the messaging system is already setup, FACTS can simply leverage this for communication.  Additionally, the server initializes an empty CCBF data structure.

\paragraph{Sending and receiving messages:}
We now describe how FACTS originates, forwards, and verifies messages.
We start our description with an auxiliary protocol $\orig(A,x)$ between a user $A$ and the server $S$ to originate a new message $x$.  This protocol is used to create an origination tag $\tag_x$ containing information about the message and originator.  This tag binds the originator's identity $A$  to the message $x$ to enable recovery upon an audit, while keeping $A$ private from receiving users, and keeping the message $x$ private from the server $S$.

Roughly, this protocol works by having $S$ produce a signature on (a hash of) the message together with the originator's identity.  Due to the use of the hash, $S$ produces this signature without learning anything about the message, while the fact that $S$ includes the originator's identity in this signature prevents a malicious originator from including the wrong identity in the message. Moreover, since the tag is bound to the message, this prevents a replay attack where an adversary reuses tags across messages to change the identity of the originator. 

\begin{algorithm}[htbp]
\begin{enumerate}
  \item To originate a message $x$, the originator $A$ chooses a random salt $r \gets \bool^\secparam$, computes a salted hash $h = H(r || x)$,
    and sends $h$ to $S$.
  \item $S$ computes an encryption of the sender's identity, $e \gets \Enc_{PK_S}(A)$, and produces signature $\s = \Sig_{SK_S}(h || e)$.  $S$ sends the tuple $(e, \s)$ to $A$.
  \item $A$ outputs $\tagg_m = (r,e,\s)$.
\end{enumerate}
\caption{$\orig(A,x)$} 
\end{algorithm}

Next, we describe the $\send$ protocol which makes use of the $\orig$ protocol to send a message $x$ between clients $A$ and $B$ while preserving (encrypted) information about the originator of $x$.  $x$ can either be a newly originated message or a forward of a previously received message.  In either case, $\send$ runs the $\orig$ protocol to produce a new tag $\tag'_x$ on the message $x$.  In the case of a new message, $\tag'_x$ is sent along with the message, while in the case of a forward, it is discarded and the message is forwarded along with its original tag instead.

\begin{algorithm}[htbp]
\begin{enumerate}
    \item If $\tagg_x=\perp$, then $x$ is a new message $A$ wants to originate.  
    $A$ runs $\tagg_x \gets \orig(A,x)$.
    \item If $\tagg_x\neq \perp$ $x$ is a message that $A$ wants to forward.  $A$ runs $\tagg'_x \gets \orig(A,x)$ and discards the output.
    \item $A$ sends $(\tagg_x,x)$ to $B$ using the E2E messaging platform's $\sendee$ protocol.
\end{enumerate}
\caption{$\send(A,B,\tagg_x,x)$}
\end{algorithm}

$\receive$ is a non-interactive algorithm that 
allows a receiving user to verify the tag, $\tag_x$, affiliated with a message $x$.  Specifically, the receiver $B$ verifies the server's signature included in $\tag_x$ to make sure that the tag indeed corresponds to $x$ and that the originator id has not been modified.  Importantly, $B$ can perform this verification without learning the identity of the originator since the tag contains an encryption of this identity (this ciphertext is what is verified by $B$).

\begin{algorithm}
\begin{enumerate}
    \item Parse $\tagg_x$ as $\tagg_x=(r,e,\s)$
    \item Compute $h=H(r||x)$
    \item Run $\Ver_{PK_S}(\s,(h||e))$ to check that $\s$ is a valid signature by the server on $(h||e)$.  If not, then discard the received message.
\end{enumerate}
\caption{$\receive(A,B,\tagg_x,x)$}
\end{algorithm}

\paragraph{Complaints and Audit:}
We now describe how FACTS allows users to complain about received messages and to trigger an audit once enough complaints are registered on a message.  For these methods we make extensive use of a CCBF data structure for (approximately) counting complaints and detecting when a threshold of complaints has been reached.

The $\complain$ protocol is used by a receiving user to issue a complaint on a received message $(\tag_x, x)$.  We assume that prior to issuing a complaint the user verifies that $\tag_x$ is valid using the $\receive$ protocol, and thus will only consider the case of valid tags.  To issue a complaint on $(\tag_x,x)$, the user $C$ calls CCBF.$\Increment(\tag_x,C)$.  As described in Section~\ref{sec:ccbf}, this runs a protocol with the server in which the user (eventually) sends the location of a bit to flip to 1 to increment the CCBF count for the message $x$.  To prevent malicious adversaries from flooding FACTS with complaints, we enforce a limit of $L$ complaints per user per epoch.  Note that since the server knows the identities of complaining users, he can easily enforce this restriction. 

Two important observations are in order here.  First, we use $\tag_x$ rather than the message $x$ as the item to increment in the CCBF.  The reason for this is that the tag is unpredictable to an adversary who has not received the message $x$ through FACTS (even if $\A$ knows $x$).  Second, we note that the CCBF.$\Increment$ procedure is inherently sequential.  It requires that the CCBF table $T$ be locked for the duration of the $\Increment$ call to prevent race condition and to maintain obliviousness (see Section~\ref{sec:ccbf} for discussion). This means that only one user can run this procedure at a time.  Thus, we focus on making this procedure as cheap as possible to minimize the impact of this bottleneck.  In case multiple clients call $\complain$ at overlapping times, the server can queue these complaints and process them one at a time.


\begin{algorithm}
\begin{enumerate}
    \item Parse $\tagg_x$ as $\tagg_x=(r,e,\s)$
    \item Call CCBF.$\Increment(C,\tagg_x)$
\end{enumerate}
\caption{$\complain(C, \tagg_x,x)$}
\end{algorithm}

The $\audit$ protocol checks whether a threshold of complaints has been reached for a given message $x$ and, if so, triggers an audit of this message.  This protocol works by using the CCBF.$\TestCount$ protocol to check whether the threshold $t$ of complaints has been reached on this message.  If this returns True, then the user simply sends $(\tag_x,x)$ to the server who first checks the validity of the tag, and then if it's valid, decrypts the corresponding part of the tag to recover the identity of the message originator.

An important observation is that the CCBF.$\TestCount$ operation is read-only and thus does not need to block.  Thus, unlike the $\complain$ command, many clients can execute the $\audit$ command in parallel.

\begin{algorithm}
\begin{enumerate}
    \item Parse $\tagg_x$ as $\tagg_x=(r,e,\s)$
    \item Call CCBF.$\TestCount(\tagg_x,t)$.
    \item If $\TestCount$ returns True, $x$ sends $(\tagg_x,x)$ to $S$
    \item $S$ verifies that the tag is valid by checking the $\s$ is a valid signature on $h||e$ where $h=H(r||x)$.  
    \item If so, $S$ recovers the identity ($A$) of the originator by computing $A=\Dec_{SK_S}(e)$.
\end{enumerate}
\caption{$\audit(C, \tagg_x,x)$}
\end{algorithm}

We note that $\audit$ allows the server to learn the message $x$ and the originator $A$.  We do not specify what the server does upon learning this information, as that is specific to a particular use of FACTS.  One possible option is for the server to review $x$ to see if it is truly a malicious message, and if so, block the user $A$ from sending further messages.  However, this decision is orthogonal to the FACTS scheme and we do not prescribe a particular action here.

\section{Experimental  Evaluations}\label{sec:experiments}
In this section, we empirically evaluate the accuracy and performance of FACTS.  We perform two sets of experiments.  The first, measures the error in terms of number of complaints above or below the threshold as a function of the total number of complaints.  The second, measures the performance overhead for messaging and complaint as a function of the threshold.

\subsection{Experimental parameters}
For our experiments, we set the maximum number of complaints per epoch $n=1,000,000$.  If we consider an epoch of one day, this results in approximately 11.6 complaints per second. To understand accuracy and efficiency of FACTS, we measure them for a range of thresholds $100 \le t \le 1000$.  With these fixed, we set the remaining parameters according to Corollary~\ref{cor:params}.  In particular, we set the server's storage $s=96n=12MB$.  The user set size $u$ varies from (approximately) 47,000 to 470,000 bits, while the message set size $v$ goes from (approximately) 740 to 7400.

\subsection{Accuracy and stability} 
To measure the accuracy of FACTS, we observe the actual number of complaints necessary to cause an audit on a single message as a function of the background noise (i.e., total complaints on other messages).  We calculate both the mean and the standard deviation of this value to capture the accuracy and stability of the complaint mechanism.  To get a statistically meaningful estimate of these, our experiments run 1000 iterations of each parameter configuration.

The results of our experiments are presented in Figure~\ref{fig:Accuracy}.  The left side of this figure shows the mean number of complaints to trigger an audit for a given threshold $t$.  As can be seen from the error bars, the absolute errors in number of complaints is quite small, with a maximum deviation of about 10 complaints at a threshold of 1000.  Not surprisingly, we see that this error increases as the background noise increases, but the mean number of complaints remains remarkably steady at the desired value.  The right side of Figure~\ref{fig:Accuracy} shows the relative standard deviation of the number of complaints as a function of background noise.  From this graph we can see that the relative error is only a few percent, with a maximum relative error of about 3.5\%.  Not surprisingly, the threshold 100 measurement incurs the highest relative error because the noise is a much higher ratio when compared to the threshold.  These experiments suggest that FACTS achieves good accuracy for a wide variety of threshold and background noise.

\begin{figure*}
\hspace*{-0.7in}
\includegraphics[width=1.2\textwidth]{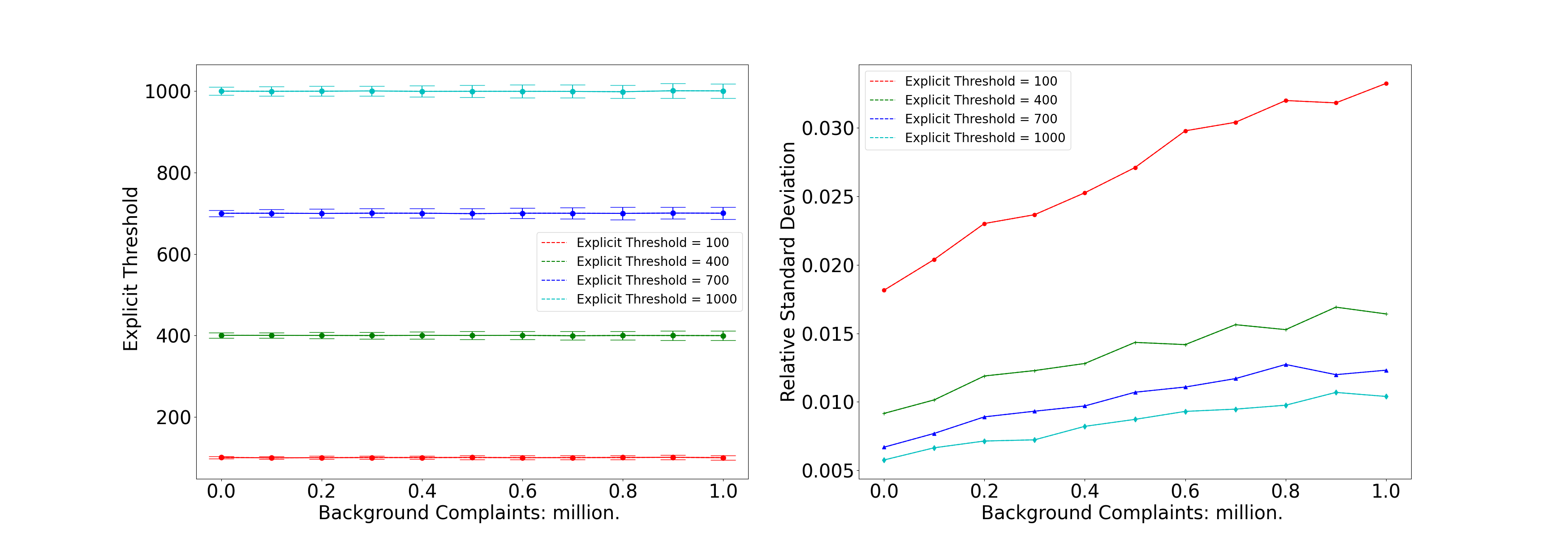}
\centering
\caption{The (left) Mean and (right) Relative Standard Deviation of Experimental Explicit Threshold vs. The Number of Background Complaints.}
\label{fig:Accuracy}
\end{figure*}



\subsection{Performance overhead}

Our next set of experiments measures the performance overhead of FACTS as a function of the threshold to start an audit.  Specifically, we measure the overhead of sending a message using FACTS, and the cost of issuing a complaint.  We note that for the message sending cost, we do not measure the cost of the EEMS communication, instead only measuring the added overhead due to FACTS.

For these experiments, we implemented both the client and server using the Rust programming language.  We used SHA-3 for a hash function, and for encryption and signatures we used a popular Rust library's implementation of OpenSSL's ChaCha20-Poly1305 protocol and Ed25519 respectively. To instantiate the CCBF, we used a simple library that allows memory to be bit addressed, rather than byte addressed, which gains us a quick, compact way to store the CCBF data structure. 

To simulate network overhead, we implemented a simple web server and client, which communicated over a (simulated) 8 Mbps network with a latency of 80ms, using TLS 1.3.  Since we are only measuring the overheads of FACTS over the underlying EEMS, our measurements 
did not include the time to send the message over the EEMS, nor the time to establish the TLS connection. All experiments were run on a 4.7Ghz Intel Core i7 with 16GB of RAM, with a sample size of 100 for each metric.  As in the accuracy experiments, we set $n=1,000,000$ and threshold varying from 100 to 1000, with the remaining parameters determined by Corollary~\ref{cor:params}.

For our measurement of message origination we looked at the cost of originating and sending a message of size 2MB. Creating and sending such a message with the encrypted hash and identity took 98ms, which indicates that the major bottleneck in this process is the 80ms network latency. We see then that when a user wishes to forward a message, they will still call $\orig(A,x)$, but then forward the original message whereas in an EEMS this would just require a forward.  Thus, the overhead of FACTS on a forward is slightly less than 100ms.

\begin{figure}
  \includegraphics[width=\linewidth]{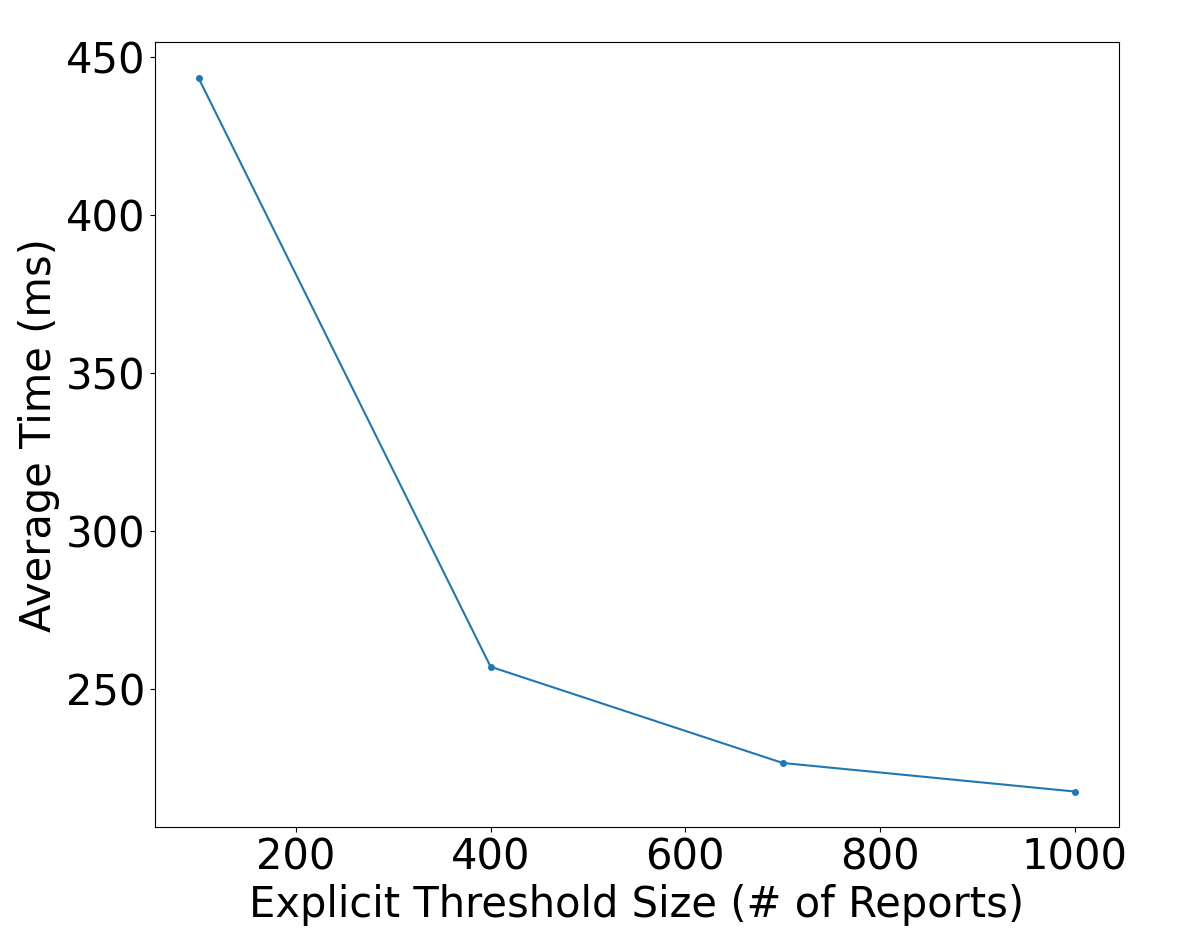}
  \caption{Complaint overhead as a function of threshold}
  \label{fig:report_msg_threshold}
\end{figure}

Figure~\ref{fig:report_msg_threshold} shows our measurements of the time to issue a complaint as a function of the audit threshold.  The time for this is dominated by the time to retrieve the user set (i.e., the bits that the user can write) from the server.  Since the size of this set $u=O(n/t)$, this time grows inversely with the threshold $t$.  Thus, as the threshold increases, the total complaint time decreases very quickly going down to essentially just the network latency when $t=1000$.

These experiments show that both the (added) cost of sending messages and the cost of complaints (for sufficiently large $t$) are dominated by the networking costs.  Thus, as long as the latency of the network is reasonable, FACTS can scale to millions of complaints per day.
\section{Security of FACTS}\label{sec:proofs}
In this section we analyze the security of FACTS.  We provide security definitions capturing the privacy and integrity guarantees provided by FACTS and prove that our protocols described in Section~\ref{sec:construction} achieve these definitions.

\subsection{Adversary Model}
We consider two different types of adversaries against FACTS.  The first is an honest-but-curious server $S$.  Such a server may also collude with some of the users.  However, all such users, as well as the server, will follow the protocol.  This adversary class models what the FACTS server learns in running the system, so we want to limit what the server learns.  However, we have to assume that the server acts honestly, as a malicious server can fully break the integrity and availability of FACTS.  For example, since the server produces the signatures binding originators to messages, a malicious adversary with knowledge of this key could arbitrarily assign originators by forging this signature.

We also consider a second type of adversary controlling a group of malicious users who do not collude with the server.  Such users may want to violate the confidentiality of FACTS by learning extra information about messages or complaints, beyond what they learn through the messages they validly receive.  Or, they may want to break the integrity of the complaint and audit mechanism of FACTS to blame innocent parties for audited messages, or to delay or speed-up the auditing of targeted messages.  This models an external adversary, say a malicious company or government, who may want to distribute fake information without being audited or may want to block certain information or users from the system.

\subsection{Privacy}
We begin by looking at the privacy guarantees provided by FACTS.

\paragraph{Privacy vs. Server:}
We first give a definition for privacy against a semi-honest server who may also collude with some semi-honest users.  In this setting we aim to argue that unless a message is audited or is received by an adversarial user, the server learns no information about the message or the complaints on the message.  In particular, the server should not be able to tell whether any message is a new message or a forward and how many, if any, complaints this message may have.  In fact, the only thing that the server learns is the \emph{metadata} of who is sending messages to whom and who is issuing complaints, but not anything more.

Specifically, 
we propose a real-or-random style definition to
capture privacy against the server.  This definition captures the fact that the view of the server (and colluding users) until a message is audited or received by a colluding user just consist of random values, and thus is independent of the messages and complaints.

Concretely, we define the following game between an adversary $\A$ controlling the server (and possibly some colluding users) and a challenger.

\begin{enumerate}
    \item[] \hspace{-2em} \underline{$\GameSP$:}
    \item The challenger runs $\setup(c)$ to set up the EEMS with $c$ clients.  He hands all keys corresponding to corrupted parties to $\A$
    \item $\A$ chooses a sequence of messages $((\mathrm{send}, A_0,B_0,\tag_{x_0}, x_0),\ldots,(\mathrm{send}, A_\ell,B_\ell,\tag_{x_\ell}, x_{\ell}))$\footnote{We note that since $S \in \A$, $\A$ can produce valid-looking tags for each of these messages by producing the necessary signatures.}, and a sequence of complaints  $((\mathrm{complain}, C_0, \tag_{x^c_0}, x^c_0), \ldots, (\mathrm{complain}, C_{\ell'}, \tag_{x^c_{\ell'}})$ and interleaves them arbitrarily. We require that none of the sending users ($A_i$), receiving users ($B_i$), or complainers ($C_i$) are controlled by $\A$.
    \item The challenger chooses $b \gets \bool$ and does the following:
    \begin{enumerate}
        \item If $b=0$, Run the $\send$ and $\complain$ protocols with inputs supplied by $\A$, giving $\A$ the resulting server view.
        \item If $b=1$, 
        \begin{itemize}
            \item for each $\send$ command, choose $r \gets \bool^\secparam$ and send this to $S$.  Choose $x' \gets \bool^{|x|+|\tag_x|}$ and send $x'$ from $A_i$ to $B_i$ using EEMS.$\sendee$.
            \item The challenger maintains a set $\USED \subseteq [s]$\footnote{Recall that $s$ is the size of the CCBF bit vector $T$}.  For each $\complain$ command, the challenger chooses $ind \gets [s] \setminus \USED$, sends $ind$ from $u_i$ to $S$, and adds $ind$ to $\USED$.
        \end{itemize}   
    \end{enumerate}

    \item $\A$ outputs a bit $b'$
    \item We say that $\A$ has advantage $$\AdvSP = |\Pr[b=b']-1/2|.$$
\end{enumerate}

\begin{definition}[Privacy vs. Server]
A FACTS scheme is \emph{private against a semi-honest server} if the adversary has a negligible advantage in the game above $\AdvSP \le \negl(\secparam)$
\end{definition}

\begin{theorem} FACTS is private against a semi-honest server
\end{theorem}

\begin{proofsketch}
First, consider the server's view on a $\send$ command.  This view consists of a message $h = H(r||m)$ for $r \la \bool^\secparam$ and the leakage from EEMS.$\sendee$, i.e., the identities $A$ and $B$, as well as $|(\tag_x,x)|$.  Since the challenger uses the same sender, receiver, and message length, the only thing left to prove is that $h$ is indistinguishable from random.  Since $r$ is chosen uniformly at random, and $H$ is a random oracle, $H(r||m)$ is uniformly random to $\A$ unless $\A$ queries $H(r||m)$.  However, since $\A$ makes at most $\poly(\secparam)$ queries to $H$, the probability that he makes this query is at most $\poly(\secparam)/2^{\secparam} \le \negl(\secparam)$.

Next, we consider the $\complain$ commands.  The server's view on a complaint consists of the complainer's ID $C$ and an index in the CCBF to flip to 1.  In a real execution of $\complain$, this index is chosen at random from the set $S_C \cap V_x$ where $S_C = \{i\in U_C \mid T[i] = 0\}$ and $V_x$ is the list of item locations for $x$.\footnote{Technically, the item used in the CCBF is the tag $\tagg_x$, but we use $x$ here for ease of notation.}  
However, since $U_C$ and $V_x$ are chosen at random, we can equivalently sample a random 0-index in the bit vector $T$ and then choose $U_C$ and $V_x$ conditioned on them containing this location.  Hence the location sent to the server is uniformly random unless $\A$ makes the corresponding $H$ query, which only happens with $\negl(\secparam)$ probability.
\end{proofsketch}

The above theorem states that, beyond the meta-data of who sent a message to whom and who has sent complaints and when, FACTS reveals no information about messages and complaints to a semi-honest server until an audit occurs (or a malicious user receives a message).  Moreover, the view of the server is completely random when conditioned on the meta-data.  Now, suppose that a message $x$ is audited (or is received by an adversary-controlled user).  When this happens, the adversary learns the tag and message $(\tag_x,x)$.  This enables $\A$ to learn the identity of the originator (by decrypting it from $\tag_x$) and to learn the entire history of this message, i.e., the transmission and complaint history of $x$.  However, since the server's view of all other messages is indistinguishable from independent random strings (modulo the meta-data), the adversary does not learn anything more about these messages as a result of an audit on $x$. 

\paragraph{Privacy vs. Users}
We now proceed to analyze security of our protocol against (possibly malicious) users that are not colluding with the server.  This models the case of a third party adversary that tries to learn information about the messages and complaints in FACTS.  Here, we no longer assume that a message $x$ is never received by a malicious user and thus we cannot use a real-or-random style definition as before.  Instead, we argue that a user cannot distinguish a new message from a forwarded message unless another corrupted user has previously seen that message.  This also shows that a malicious user cannot learn the identity of the message originator.  Since users do not receive any communication on complaints, we only consider message privacy here.

Concretely, we define the following game between an adversary $\A$ controlling a set of users, and a challenger.

\begin{enumerate}
    \item[] \hspace{-2em} \underline{$\GameCP$:}
    \item The challenger runs $\setup(c)$ to set up the EEMS with $c$ users and gives all key material for the corrupted users to $\A$. Let $B \in \A$ be a user controlled by the adversary.
    \item $\A$ chooses messages $x, x'$ s.t. $|x|= |x'|$ and honest users $O, A \notin \A$
    \item The challenger chooses $b \mid \bool$ and does the following:
    \begin{enumerate}
        \item If $b=0$, the challenger runs $\send(O,A,\perp,x')$ and $\send(A,B,\perp,x)$ with $\A$ receiving the view of $B$.
        \item If $b=1$, the challenger runs $\send(O,A,\perp,x)$ and $\send(A,B,\tag_x, x)$ (where $\tag_x$ is the tag received by $A$ from $O$).
    \end{enumerate}
    \item $\A$ outputs a bit $b'$
    \item We say that $\A$ has advantage $$\AdvCP = |\Pr[b=b']-1/2|.$$
\end{enumerate}

\begin{definition}[User privacy]
A FACTS scheme achieves \emph{privacy against malicious users} if the adversary has a negligible advantage in the game above $\AdvCP \le \negl(\csecparam)$
\end{definition}

\begin{theorem}
FACTS achieves privacy against malicious clients.
\end{theorem}

\begin{proofsketch}
The view of $B$ on an execution of $\send(\cdot,B,\tag_x,x)$ consists of the received message and tag $(\tag_x, x)$ where $\tag_x = (r,e,\s)$.  Since $e$ is a semantically secure encryption of the identity of the originator, $\A$ cannot distinguish between the case when $e = \Enc(A)$ (when $b=0$) and the case when $e = \Enc(O)$ (when $b=1$) except with advantage negligible in $\csecparam$.  Additionally, since $\tag_x$ is generated identically both when $b=0$ and $b=1$ except for this change in $e$, this means that $\tag_x$ does not help $\A$ distinguish between these two cases.
\end{proofsketch}

\subsection{Integrity}
We now turn to the integrity guarantees provided by FACTS.  We aim for a few different notions of integrity to show that malicious users cannot interfere with the complaint and audit process.  First, no adversary controlling a subset of the users should be able to frame an honest user as the originator of an audited message he did not originate.  Second, an adversary controlling a subset of the users should not be able to significantly delay the audit of a malicious message.  In particular, such an adversary should not be able to prevent a malicious message sent by one of his users from being audited.  Finally, an adversary controlling a small set of users should not be able to significantly speed up the auditing of a targeted message.  In particular, such an adversary should not be able to cause an audit without complaints from some honest users.

We begin by defining the following game between a challenger and an adversary $\A$ controlling a subset of the users to capture the inability of an adversary to forge a valid tag that it has not seen before.

\begin{enumerate}
    \item[] \hspace{-2em} \underline{$\GameUF$:}
    \item The challenger runs $\setup(c)$ to set up the EEMS with $c$ clients and gives all key material for the corrupted clients to $\A$.
    \item $\A$ requests $\send$ operations on messages of its choice both from honest and corrupted clients. ($\A$ is given the view of corrupted clients in all these executions consisting of $(\tagg_x,x)$.)
    \item $\A$ outputs a tag, message pair $(\tagg_y, y)$
    \item We say that $\A$ WINS if $\tagg_y$ is a valid tag for message $y$ with originator $O \notin \A$, and there has not been a prior command $\send(O,\cdot,\perp,y)$.
\end{enumerate}

\begin{definition}[No framing]
We say that a FACTS scheme disallows \emph{framing} if for any PPT $\A$, $\A$ WINS in the above game with probability at most $\negl(\csecparam)$.
\end{definition}
\begin{theorem}
The FACTS scheme is unforgeable.
\end{theorem}
\begin{proofsketch}
A valid tag $\tagg_y$ with originator $O$ consists of $\tagg_y = (r,e,\s)$ where $r$ is a random seed s.t. $H(r||y) = h$, $e = \Enc_{PK_S}(O)$, and $\s = \Sig_{SK_S}(h||e)$.  Thus, to frame $O$, $\A$ needs to produce a valid signature on $h||\Enc(O)$.  $\A$ can observe tags from polynomially many messages originated by $\O$, but except with probability negligible in $\secparam$ none of them will have the same value $h$.  Thus, by the unforgeability of $\Sig$, $\A$ cannot produce the necessary signature except with probability negligible in $\csecparam$.
\end{proofsketch}

Next, we give a definition that captures the ability of an adversary controlling a subset of the clients to delay the audit of a particular message.  Our goal is to show that the adversary cannot protect a malicious message from being audited.

Specifically, we define the following game,

\begin{enumerate}
    \item[] \hspace{-2em} \underline{$\GameND$:}
    \item The challenger runs $\setup$ to set up the EEMS with $c$ clients and gives all key material for the corrupted clients to $\A$.
    \item $\A$ issues a single $\send(A,B,x)$ command with $A \in \A$ to produce $\tag_x$
    \item $\A$ outputs a list of $\complain$ commands with at most $n$ total complaints, of which at least $\ell$ are complaints on $\tag_x$.
    \item The challenger runs the specified complaint commands, and then runs $\audit(A,\tag_x,x)$
    \item We say that $\A$ WINS if this audit is not successful (i.e., the audit threshold is not reached). 
\end{enumerate}

\begin{definition}[No delay]
We say that a FACTS scheme is \emph{$\ell$-audit delay resilient} for integer $\ell < n$ if for any PPT $\A$, $\A$ WINS in the above game with probability at most $\negl(\secparam)$.
\end{definition}

\begin{theorem}
The FACTS scheme is $\ell$-audit delay resilient for any $\ell \ge 1.1t+.4\secparam + .7\sqrt{\secparam t}$.
\end{theorem}

\begin{proofsketch}
This follows immediately from Theorem~\ref{thm:fn}
\end{proofsketch}

Next, we define the following game to capture the 
ability of a small number of malicious users to cause the audit of some message.  Importantly, this definition also captures the case where malicious users try to audit an honest message (on which there are no complaints by honest users).  Specifically, the following game is between an adversary $\A$ corrupting at most $\ell$ users and a challenger

\begin{enumerate}
    \item[] \hspace{-2em} \underline{$\GameNSU$:}
    \item The challenger runs $\setup$ to set up the EEMS with $c$ clients and gives all key material for the $\ell$ corrupted clients to $\A$.
    \item The challenger runs a single $\send(A,B,x)$ command for $A \notin \A$ and $B \in \A$.
    \item $\A$ may issue at most $L$ $\complain$ commands per each user he controls.\footnote{Recall that FACTS enforces a limit of $L$ complaints per user per epoch.}
    \item The challenger runs the specified $\complain$ commands, and then runs $\audit(\cdot, \tagg_x,x)$.
    \item We say that $\A$ WINS if this audit is successful.
\end{enumerate}

\begin{definition}[No speed up]
We say that a FACTS scheme is \emph{$\ell$-party audit speed-up resilient} if for any PPT $\A$ controlling at most $\ell$ users, $\A$ WINS in the above game with probability at most $\negl(\secparam)$.
\end{definition}

\begin{theorem}
The FACTS scheme is $\ell$-party audit speed-up resilient for $\ell \le (t - \Cminfp\sqrt{\lambda t})/L$.
\end{theorem}

\begin{proofsketch}
This follows immediately from Theorem~\ref{thm:fp} because each user $\in \A$ makes at most $L$ complaints.
\end{proofsketch}

\section{Alternative FACTS} \label{sec:altFACTS}
In this section we describe several optimizations or enhancements to the basic FACTS protocol.

\paragraph{Throttling complaints.}
The FACTS system and underlying CCBF data structure assume a global limit $n$
on the number of complaints per epoch, but do not require any per-user limit
besides the natural limit of $u$, the size of the user set.

However, there is some potential for abuse by users who issue many complaints in
a single epoch: they may attempt to ``attack'' another known message by issuing
multiple complaints that set bits in that message's user set;
they may collude with others and attempt to go over the total per-epoch limit of
$n$ complaints; or they may simply attempt a denial-of-service attack to prevent
other complaints from being issued.

An simple solution to these problems is to apply a limit $\ll u$ on the maximum
number of complaints per user per epoch.  This is easy for the server to apply,
since users are authenticated during the \complain{} protocol. More nuanced limits
based on a user's reputation or longevity on the platform could also be applied.

Users with a small ``quota'' of allowed complaints per epoch could even be encouraged
to participate initially in the complaint process by forwarding questionable content
to a trusted reputable user on the system, who would then presumably apply their
own judgment and possible issue a complaint in turn. This idea is aligned with
many existing content moderation settings on (unencrypted) social media platforms.

\paragraph{Regional complaint servers.}
The most significant performance bottleneck in FACTS is the necessary global
lock on the table $T$ while a single user is waiting to download their
user set $U_C$ and reply with their complaint index. Even though the communication size
is quite small for practical settings, the inherent latency across global communications
networks may impose a challenge.

For example, if many complaining users have a round-trip latency of more than 200ms,
then the global complaint rate among all users cannot be higher than 5 complaints per
second, or some 432,000 complaints per day, regardless of any parameter settings or
chosen epoch length.

One possible solution for a large-scale platform facing this issue would be to
allow multiple local complaint servers, each with their
own CCBF table $T$, to independently operate and accumulate complaints per messages.
This makes sense, as most targeted misinformation content is local to a given country
or region, and it would still be possible for each regional server to share audited
message information with others in order to prevent spread of viral false content
between regions.

\paragraph{Third-party audits.}\label{sect:third-party-audits}
While many messaging and social media platforms currently employ their own
``in-house'' teams for content moderation, there have been some attempts at
separating the role of the server from that of auditor.

From a protocol standpoint, we can imagine a separate Server and Auditor: the former
is semi-honest, handles the encrypted messaging system and maintains the public
CCBF table $T$. The Auditor is fully honest and non-colluding,
but computationally limited; intuitively, the third-party Auditor should only
be involved once a messaged has passed the desired threshold of complaints.

The FACTS system supports this option easily with the need for any additional
cryptographic setup during origination. Because the CCBF table $T$ is globally shared
among all users as well as the Auditor, any complaining user who computes
$\TestCount{}$ on their own to see that the probabilistic threshold has been surpassed,
can then forward their complaint (i.e., the opened message) directly to the Auditor.
Being fully honest, the Auditor may hold a copy of the decryption key from origination
and use this to determine what kind of action may be necessary (such as suspending
the originating user's account, flagging the message, etc.).

While it doesn't appear idea imposes any additional interesting challenges from a cryptographic
standpoint, it could be useful for some kinds of messaging platforms.

\paragraph{Hiding message metadata.}
Our FACTS system is certainly no more private than the underlying EEMS
which is being used to actually pass messages between users. In our analysis,
we explicitly assumed that the EEMS leaks metadata on the sender and recipient of
each message, but not the contents.

However, some existing EEMS attempt to also obscure this metadata in transmitting
messages, so that the server does not learn both sender and recipient of any
message. This can trivially be accomplished by foregoing a central server and
doing peer-to-peer communication (note that FACTS may still be useful as a central
complaint repository); or using more sophisticated cryptography to
hide metadata \cite{stadium,verdict,ricochet}.

Of particular interest for us is the recently deployed \emph{sealed sender}
mechanism on the popular Signal platform \cite{signal-ss}. The goal in this
case is to obscure the sender, but not the recipient, from the server handling
the actual message transmission. We note that this concept plays particularly
well with FACTS, as the additional leakage in our protocol of the identity
of each complaining user, can be presumably correlated via timings with the
receipt of some message, but this is exactly what is revealed under sealed sender
already! Both systems thus work to still hide message sender and originator identities
(at least until an audit is performed).

However, note that recent work \cite{signal-ndss21} has shown that some timing
attacks are still possible under sealed sender, and the same attacks would apply
just as well to FACTS. But the solutions proposed in \cite{signal-ndss21} might also
be deployed alongside FACTS to prevent such leakage; we leave the investigation of
this question for future work.



\section{Related Work}\label{sec:related}
\paragraph{Message Franking:}
The most common approach today for reporting malicious messages in encrypted messaging systems is \emph{message franking}~\cite{C:DGRW18, C:TGLMR19, C:GruLuRis17}.  Message franking allows a recipient to prove the identity of the sender of a malicious message.  However, message franking is focused on identifying the last sender of a message, whereas we are interested in identifying the originator.  Moreover, message franking does not provide any threshold-type guarantees to prevent unmasking of senders given only one (or a few) complaints.

\paragraph{Oblivious RAM (ORAM):}
Oblivious Random-Access Memory (ORAM)~\cite{STOC:Goldreich87, STOC:Ostrovsky90, JACM:GolOst96} allows a client to obliviously access encrypted memory stored on a server without leaking the access pattern to the server.
The standard ORAM definition assumes a single user with full control over the database.
While some important progress has been made on multi-client ORAM protocols
\cite{maffei-msmcoram,concuroram,MOSE},
these solutions are still not scalable to millions of malicious users as would be needed
for our application.

\paragraph{Oblivious Counters and Oblivious Data-Structures:}
Like CCBF, oblivious counters~\cite{EC:KatMyeOst01, FC:GohGol05} build counters that can be stored and incremented without revealing the value of the counter.  However, these techniques focus on exact counting, and do not provide efficient ways for storing large numbers of counters, as needed for our applications.  More generally, oblivious data-structures, e.g.~\cite{CCS:WNLCSS14, AC:KelSch14, SP:Shi20} 
construct higher-level data structures such as heaps, trees, etc. to enable oblivious operations over encrypted data.  However, these largely focus on higher-level applications and do not provide the compression achieved by CCBF.

\paragraph{Privacy-Preserving Sketching:}
CCBF can be viewed as a small data structure (a \emph{sketch}) for storing the counts of complaints on a large set of messages.  There has indeed been a lot of recent interest (e.g.,~\cite{PoPETS:CDKY20, CCS:WJSYG19, NDSS:MelDanDec16, CCS:JanJoh16, CCS:FMJS17, AM14, FGJ+15}) in private sketching algorithms for cardinality estimation, frequency measurement, and other approximations.  However, these works generally focus on a multi-party setting, with multiple parties running secure computation to evaluate the statistic in question.  Since our goal was to restrict ourselves to user-server communication only, such techniques do not seem applicable to our setting.

\bibliography{bib,CryptoBib/abbrev2,CryptoBib/crypto}
\bibliographystyle{plain}

\appendix

\section{Proofs for tail bound probabilistic analysis}

We now complete the proofs of lemmas and theorems in \cref{sec:tailbounds}.

We start with the following standard way to approximate numbers near
1 with exponentials.
\begin{lemma}\label{lem:expapprox}
    For any real constant $\alpha>0$, and any
    real $x$ with $0<x\le\alpha$, we have
    \[
    \exp\left(-\tfrac{1}{\alpha}\ln\tfrac{1}{1-\alpha}\cdot{}x\right)
    \le 1-x < \exp(-x).
    \]
\end{lemma}

We also re-state this straightforward consequence of the Hoeffding/Chernoff bound on the sum of random variables:
\begin{lemma}\label{lem:hoeff}
    Let $X_1,\ldots,X_n$ be independent Poisson trials, and write $Y=\sum_i X_i$ for their sum. If $\Expect[Y] = \mu$, then for any $\delta>0$, each of
    \(\Pr(Y \ge \mu + \delta)\) and  \(\Pr(Y \le \mu - \delta)\)
    are at most \(\exp(-2\delta^2/n)\).
\end{lemma}

We now recall and prove the building-block lemmas from \cref{sec:tailbounds}.

\lemmalowerp*

\begin{proof}
    From \eqref{eqn:pw}, we know this probability is exactly
    $p_w = 1 - \frac{\falling{(s-u)}{w}}{\falling{s}{w}},$
    where $w = v-\tau$ is the number of unfilled slots remaining.
    Using \cref{lem:expapprox} we have
    $\frac{\falling{(s-u)}{w}}{\falling{s}{w}}
    \le \left(1-\frac{u}{s}\right)^w
    \le \exp(-uw/s)$,
    which means that
    \[p_w \ge 1 - \exp\left(-\tfrac{u(v-\tau)}{s}\right)
        = 1 - \exp\left(-\tfrac{u\tau}{s}\cdot{}\left(\tfrac{v}{\tau}-1\right)\right).
    \]
    Applying the two lower bounds on $\frac{u\tau}{s}$ and $\tfrac{v}{\tau}$
    from the lemma statement yields the claimed result.
\end{proof}

\lemmaupperp*

\begin{proof}
    Using again \eqref{eqn:pw}, the probability is exactly
    $p_w = 1 - \frac{\falling{(s-u)}{w}}{\falling{s}{w}},$
    where again $w\le v$ is the number of unfilled slots for
    item $x$. Then
    \[
        \tfrac{\falling{(s-u)}{w}}{\falling{s}{w}}
        \ge \tfrac{\falling{(s-u)}{v}}{\falling{s}{v}}
        \ge \left(\tfrac{s-u-v+1}{s-v+1}\right)^v
        > \left(1 - \tfrac{u}{s-v}\right)^v.
    \]
    Using upper bounds on $\frac{v}{s}$ and $u$ from the lemma statement, we have
    \[p_w
        < 1 - \left(1 - \tfrac{u}{s-v}\right)^v
        \le 1 - \left(1 - \CzeroUpper{}\tfrac{1}{v}\right)^v.\]
    Finally, the lower bound on $v$ from the lemma statement shows $\CzeroUpper{}/v \le \CboxUpper{}$, and so we can finally use the lower exponential bound of \cref{lem:expapprox} to obtain the stated result.
\end{proof}

\lemmauppertau*

\begin{proof}
    The tipping point $\tau$ is the expected number of slots filled in the table
    if $t$ of the $m$ total calls to \Increment{} were actually called on this
    particular item.
    
    We can divide the calls to \Increment{} into two groups: the $t$ calls for
    item $x$, and the $m-t$ calls for other items. The expected number of slots
    within $x$'s item set filled by the first group is at most
    $\Cphigh t$, from \cref{lem:upperp}.
    
    For the second group, these calls to
    \Increment{} on unrelated items are distributed uniformly at random among
    all table indices, and so their expected fraction within this item set
    is the same as their overall fraction in the table. Therefore, the expected
    number of slots filled by calls to \Increment{} on other items is at most
    \[\frac{(m-t)v}{s} < \frac{mv}{s} 
        \le \frac{\CmaxvOvert}{\CminsOvern} t.\]
    
    By linearity of expectation, we can sum these two to obtain an upper bound
    on the total expected tipping point as given in the lemma statement.
\end{proof}

Now we can proceed to the proofs of the main theorems on the accuracy of the CCBF.

\theoremfp*

\begin{proof}
    Let $\tau_t$ be the tipping point for any actual number $m\le n$ of total set
    bits in the table $T$ and for the given threshold $t$.
    And consider random variables $X_1,\ldots,X_v$ for the $v$ slots assigned to
    item $x$, where each $X_1$ is 0 or 1 depending on whether the corresponding
    slot in table $T$ is 0 or 1. We want to know the probability that the sum
    of the $X_i$'s is at least $\tau_t$, which is what would cause $\TestCount(x,t)$
    to produce a false positive.
    
    Let $k = t - \Cminfp\sqrt{\lambda t}$ be the actual number of calls to
    \Increment{} on item $x$, and write $\tau_k$ for the tipping point at threshold
    $k$. By definition and the exact calculations for $\tau_k$ outlined earlier,
    we know that $\Expect[\sum X_i] = \tau_k$.
    
    The difference between these two tipping points $\tau_t - \tau_k$ is
    the expected number of extra slots filled by $t-k$ calls to \Increment{},
    which from \cref{lem:lowerp} is at least 
    \[\Cplow(t-k) = \Cplow{}\cdot{}\Cminfp{}\sqrt{\lambda t} \ge \sqrt{\tfrac{\lambda v}{2}},\]
    where in the last step we used the upper bound on $v$ from the assumptions
    of the theorem.
    
    The variables $X_i$ are not independent, but
    they are \emph{negatively correlated}, meaning that the whenever one slot is filled,
    it only decreases the likelihood that another is filled; intuitively, this is
    because there are now fewer chances to fill the other slot.
    Therefore we can apply the Hoeffding bound in this direction
    (\cref{lem:hoeff}) to say that
    \[Pr\left(\sum X_i \ge \tau_k + \sqrt{\tfrac{\lambda v}{2}}\right)
    \le \exp(-\lambda),\]
    as required.
\end{proof}

\theoremfn*

\begin{proof}
   Writing $k$ for the actual number of complaints given in
   \eqref{eqn:fn}, we need a tail bound on the probability that,
   after $k$ calls to $\Increment{}$ on the same item $x$, there are still
   fewer than $\Cfour{}t$ slots of $x$'s item set filled in, where the latter
   constant comes from applying the upper bound on the tipping point from
   \cref{lem:uppertau}.
   
   For this, we need a lower bound on the \emph{expected} number of
   bits set after $k$ calls to \Increment{} on item $x$; from \cref{lem:lowerp}
   this is at least $\Cplow{}k$.
   
   Now we can apply the Hoeffding bound (\cref{lem:hoeff}),
   with $\mu = \Cplow{}k$ and $\mu+\delta=\Cfour{}t\ge\tau$
   to see that the probability that less than $\tau$ bits of $x$'s user set
   are flipped is at most
   \begin{align*}
    &\exp\left(-2(\Cfour{}t-\Cplow{}k)^2/k\right) \\
    \le& \exp\left(-2\left(0.38\lambda + 0.66\sqrt{\lambda t}\right)^2/k\right)\\
    \le& \exp\left(-\frac{.28\lambda^2 + \lambda\sqrt{\lambda t} + .87\lambda t}%
        {.4\lambda + .7\sqrt{\lambda t} + 1.1t}\right) \\
    \le& \exp(-.7\lambda)  \le 2^{-\lambda}.
   \end{align*}
\end{proof}

\end{document}